\documentclass{llncs}
\usepackage{makeidx}  
\usepackage[ruled,linesnumbered,commentsnumbered]{algorithm2e}

\usepackage{color,soul}
\usepackage{amssymb,amsmath}
\usepackage{graphicx}
\usepackage{color}
\usepackage{soul}
\usepackage{tikz}
\usetikzlibrary{arrows,snakes,shapes}

\newcommand{\ie}{\textit{i.e.\,}}

\newcommand{\liststpaths}{\ensuremath{\mathtt{list\_induced\_paths}_{s,t}}}

\newcommand{\path}{\ensuremath{\mathtt{path}}}
\newcommand{\polylog}{\ensuremath{\mathrm{polylog\,}}}

\newcommand{\setofinducedpaths}{\ensuremath{\mathcal{P}}}
\newcommand{\setofinducedcycles}{\ensuremath{\mathcal{C}}}

\newcommand{\good}{N^{\textit{good}}}

\newtheorem{fact}{Fact}

\begin{document}

\title{Amortized $\tilde{O}(|V|)$-Delay Algorithm for Listing
  Chordless Cycles in Undirected Graphs\thanks{GS and MFS were
    partially supported by the ERC programme FP7/2007-2013 / ERC
    grant agreement no. [247073]10, and the French project
    ANR-12-BS02-0008 (Colib'read).}}

\titlerunning{Listing Chordless Cycles}  
%

\author{Rui Ferreira\inst{1}, Roberto Grossi\inst{2}, Romeo
  Rizzi\inst{3} \and Gustavo Sacomoto\inst{4,5} \and Marie-France
  Sagot\inst{4,5}}
%
\tocauthor{}
\institute{
  Microsoft Bing, UK
  \and
  Universit\`a di Pisa, Italy
  \and
  Universit\`a di Verona, Italy
  \and
  INRIA Grenoble Rh\^one-Alpes, France 
  \and
  UMR CNRS 5558 - LBBE, Universit\'e Lyon 1, France  
}

\maketitle    

\begin{abstract} 
  Chordless cycles are very natural structures in undirected graphs,
  with an important history and distinguished role in graph
  theory. Motivated also by previous work on the classical problem of
  listing cycles, we study how to list chordless cycles. The best
  known solution to list all the $C$ chordless cycles contained in an
  undirected graph $G=(V,E)$ takes $O(|E|^2 + |E| \cdot C)$ time. In
  this paper we provide an algorithm taking $\tilde{O}(|E| + |V| \cdot
  C)$ time. We also show how to obtain the same complexity for listing
  all the $P$ chordless $st$-paths in $G$ (where $C$ is replaced by
  $P$).
\end{abstract}


\section{Introduction}
\label{sec:introduction}

A \emph{chordless (induced) cycle} $c$ in an undirected graph $G$ is a
cycle such that the subgraph induced by its vertices contains exactly
the edges of $c$. A chordless cycle is called a \emph{hole} when its
length is at least 4. Similarly, a \emph{chordless (induced) path}
$\pi$ in $G$ is such that the subgraph of $G$ induced by $\pi$
contains exactly the edges of $\pi$. Both chordless cycles and paths
are very natural structures in undirected graphs with an important
history, appearing in many papers in graph theory related to chordal
graphs, perfect graphs and co-graphs
(e.g. \cite{Seinsche1974,ConfortiR92a,Chudnovsky2006}), as well as
many NP-complete problems involving them
(e.g. \cite{Chen2007,Haas2006,Kawarabayashi08}).

In this paper we consider algorithms for listing chordless cycles and
$st$-paths in an undirected graph $G = (V,E)$, with $n=|V|$ vertices
and $m = |E|$ edges, motivated by the algorithms for listing cycles
and $st$-paths that have been produced by an active area of research
since the early 70s \cite{Read1975,Syslo1981,BirmeleFGMPRS13}.

In this paper we present an algorithm for listing all the $C$
chordless cycles in an undirected graph $G=(V,E)$ in $\tilde{O}(m + n
\cdot C)$ time, hence with an amortized $\tilde{O}(n)$ time delay,
where $\tilde{O}(f(n,m))$ is used as a shorthand for $O(f(n,m) \,
\polylog n)$.  We also show that the same algorithm may be used to
list all the $P$ chordless $st$-paths in $\tilde{O}(m + n \cdot P)$
time, hence amortized $\tilde{O}(n)$ time delay.

There are very few algorithms in the literature for listing chordless
cycles and/or paths, where some of them have no guaranteed performance
\cite{Sokhn2012,Wild08}. The most notable and elegant listing
algorithm is by Uno~\cite{Uno03}, with a cost of $O(m^2 + m \cdot C)$
time for chordless cycles and $O(m^2 + m \cdot P)$ time for
chordless $st$-paths, hence amortized $O(m)$ time delay.

\section{Preliminaries}
\label{sec:preliminaries}

Our graphs are finite, undirected, and \emph{simple}, \ie without
self-loops or parallel edges.  Given a graph $G=(V,E)$ with $n=|V|$
vertices and $m=|E|$ edges, our task is to list out fast all its
chordless cycles.  We hence assume that $G$ is connected.  Given
$V'\subseteq V$, we denote by $E\langle V' \rangle := \{ uv \in E \mid
u,v\in V'\}$ the set of those edges which are contained in $V'$.  A
graph $G'=(V',E')$ is called a \emph{subgraph} of $G$ if $V'\subseteq
V$ and $E'\subseteq E$.  The subgraph $G'$ is called \emph{induced}
(or \emph{chordless}) if $E' = E\langle V' \rangle$.  For any
$V'\subseteq V$, we denote by $G[V'] := (V', E\langle V' \rangle)$ the
\emph{subgraph of $G$ induced by $V'$}.  Given $e\in E$, we denote by
$G\setminus e := (V,E \setminus \{e\})$ the subgraph obtained from $G$
by \emph{deleting} the edge $e$.  Given $v\in V$, we denote by
$G\setminus v := G[V\setminus \{v\}]$ the subgraph obtained from $G$
by first deleting all the edges incident to $v$, and then removing the
isolated vertex $v$.  Given a vertex $u \in V$, we denote by $N_G(u)
:= \{v\in V \mid uv\in E\}$ the \emph{neighbourhood} of $u$, the
subscript is omitted whenever the graph is clear from the context.

A \emph{cycle} is a connected graph in which every vertex has
degree~$2$.  A \emph{path} is a connected graph in which every vertex
has degree~$2$ except for two degree-$1$ vertices, $s$ and $t$, called
the \emph{endvertices} of the path. This is also called an
$st$-\emph{path} and denoted by $\pi_{st}$. Indeed, when building a
path from $s$ to $t$ edge after edge, it will be most natural, and
more precise, to think like we are orienting the traversed edges.  For
this reason, we will also write $(u,v)$ for an edge that, when
building a path, has been traversed from $u$ to $v$.

A (chordless) path (or cycle) of $G$ is a (chordless) subgraph of $G$
which is a path (or cycle).  We denote by $\setofinducedcycles(G)$ the
set of all chordless cycles in $G$.  We denote by
$\setofinducedpaths(G)$ (by $\setofinducedpaths_{st}(G)$) the set of
all chordless paths ($st$-paths) in $G$.  When $s=t$, we get those
cycles visiting $s$.  We refer to a path $\pi \in
\setofinducedpaths(G)$ by its natural sequence of vertices or edges.
A \emph{hole} is a chordless cycle of size at least $4$. Thus
$\setofinducedcycles(G)$ comprises holes and triangles.  Since there
are at most $mn$ triangles, our algorithm can be used to list the
holes of $G$ in $\tilde{O}(n)$ time each, with an overall
$\tilde{O}(mn^2)$ additive time cost.

Uno~\cite{Uno03} proposed an algorithm that lists each chordless cycle
in an undirected graph $G=(V,E)$ in $O(m)$ time while using $O(m)$
space. The first step is the following reduction to the problem of
enumerating the chordless $st$-path in a graph $G$. Based on the fact
that for any vertex $s \in V$ the chordless cycles in $G \setminus s$
are also chordless cycles in $G$, the algorithm proceeds by listing
all chordless cycles passing through $s$; and repeating the process in
$G \setminus s$, until the graph is empty. Then, to list all chordless
cycles passing through $s$ in $G' = G$, the algorithm follows the
approach of listing the chordless paths $s \leadsto t$ in $G'
\setminus (s,t)$, for each $t \in N_{G'}(s)$; and to avoid
duplications, at the end of each iteration the graph is updated to $G'
= G' \setminus t$.

Given a previously computed chordless $st$-path $\pi = v_0v_1 \ldots
v_l$, Uno's algorithm identifies the set of vertices $U \subseteq V$
such that each $u \in U$ is adjacent to some $v_j \in \pi$, and the
edge $(v_j,u)$ is contained in a chordless $st$-path $\pi' \neq \pi$
extending the prefix $\pi_j = v_0v_1\ldots v_j$. The algorithm is
kick-started by taking a shortest $st$-path (as a shortest path has
the property of also being a chordless path) and employs a recursive
strategy of vertex removal to avoid listing the same chordless path
multiple times. This ensures that each chordless path is listed once.
Uno's algorithm takes $O(m)$ time to compute $U$ and prepare the
recursive calls before it either outputs a new path or stops. The
total time is therefore $O(m^2 + m \cdot |\setofinducedcycles(G)|)$.

\section{Our Approach and Key Ideas}
\label{sec:approach-key-ideas}

We outline the main ideas which allow us to reduce the amortized cost
for a chordless cycle from $O(m)$ to $\tilde{O}(n)$, giving a total
$\tilde{O}(m + n \cdot |\setofinducedcycles(G)|)$ time to list all the
chordless cycles. Our approach relies on a variant of the
\emph{cleaning} operation introduced in~\cite{ConfortiR92a} to
recognize linear balanced matrices and even holes in
graphs~\cite{ConfortiCKV94,ConfortiCKV97}.

\subsection{Certificates for chordless $st$-path} 
\label{sub:certificate}
A listing algorithm usually takes the form of a recursive procedure
exploring the space of all solutions.  A key idea employed since the
first listing papers~\cite{Read1975} is to check for the existence of
at least one solution before branching, \ie before partitioning the
solution space in subspaces to be assigned to the children.  This
avoids unproductive recursive calls, \ie calls that do not list any
solution and whose overhead cost could completely dominate the cost of
reporting the solutions (\textit{e.g.} see~\cite{Uno97}).  In a
previous work~\cite{BirmeleFGMPRS13}, we stressed the notion of
certificate since, in a more refined recursive scheme, passing a
certificate of existence as an extra parameter may facilitate the work
of the children which may avoid running the existence check: if they
have a single child, they could be done by just passing the
certificate received as an input or a small adaptation of it.  We also
saw that more structural facts around the certificate could be useful.
For the case of $st$-paths~\cite{BirmeleFGMPRS13}, the certificate is
a DFS tree rooted in $s$ and reaching $t$, which contained an
$st$-path and also helped in other ways.  Until now, the certificate
was itself a solution or explicitely contained one.

Here we try out something new: what if our certificate guarantees the
existence of a solution but is \emph{not} itself a solution?  The
following fact suggests that the certificate for the existence of a
chordless $st$-path might be just any $st$-path.

\begin{fact} 
  \label{fact:path_induced}
  Given two vertices $s, t$ in $G$, there is a \emph{chordless} $st$-path
  in $G$ iff there is an $st$-path in $G$.
\end{fact}
Thus we allow for certificates which are somewhat less refined than
actual solutions, in the same spirit that a binary heap demands a less
strict and lazy notion of order.  This is a new asset of the notion of
certificate and opens up new possibilities.

\subsection{From chordless cycles to chordless $st$-paths} 
\label{sub:duplication}
Uno~\cite{Uno03} shows how to reduce listing chordless cycles in a
graph to listing chordless $st$-paths for all edges $(s,t)$ chosen in
a specific order (see Section~\ref{sec:preliminaries}), which is
necessary to avoid duplications in the output. The initialization step
for each edge $(s,t)$ takes $O(m)$ time as it requires to find one
chordless $st$-path. This gives the $m^2$ term in the total cost of
$O(m^2 + m \cdot |\setofinducedcycles(G)|)$ for chordless cycles.

We observed in Section~\ref{sub:certificate} that any $st$-path will
suffice as a starter, as they are our certificates of choice.  This
makes a difference for the above reduction, since using dynamic graph
connectivity algorithms \cite{Kapron2013}, it is possible to maintain
a spanning tree in $O(\polylog n)$ time per edge deletion, perform
connectivity queries in $O(\polylog n)$, and more importantly obtain
an $st$-path in $\tilde{O}(n)$. It is worth noting that it is not
known how to obtain a chordless $st$-path faster than $O(m)$. Hence,
we first build the dynamic connectivity structure as preprocessing
step. Then, for each edge $(s,t)$, in the same order as Uno's
reduction, we list the chordless $st$-paths. Before calling our path
listing algorithm for edge $(s,t)$, we test if $s$ and $t$ are
connected (Fact~\ref{fact:path_induced}): if so, we call our path
listing algorithm, paying $\tilde{O}(n)$ to find one initial
$st$-path; otherwise, we skip the edge $(s,t)$ and take the next in
order. As a result, the total initialization cost is $\tilde{O}(m+kn)$
for all edges instead of $O(m^2)$, where $k$ is the number of edges
for which we find one initial $st$-path. Note that $k \leq
|\setofinducedcycles(G)|$ as each of them surely gives rise to a
chordless $st$-path whence to a distinct chordless cycle. We obtain in
this way an $\tilde{O}(m + n \cdot |\setofinducedcycles(G)|)$ time
algorithm to list chordless cycles, if we can list $st$-paths in
amortized $\tilde{O}(n)$ time each.

\subsection{Difficulty of cleaning $st$-paths} 
\label{sub:cleaning}
Given any $st$-path, as stated in Fact~\ref{fact:path_induced} we can
\emph{clean} it to obtain a chordless $st$-path in a greedy fashion:
start from $u=s$ and iteratively take a neighbour of~$u$ that is
closest to $t$ along the path. The process stops when $u=t$. The
vertices taken in this way form a chordless path. The problem is that
the cost of such a greedy traversal of the path is upper bounded by
the sum of the degrees of the vertices along it. Unfortunately, this
sum could be $\Theta(m)$ in the worst case.

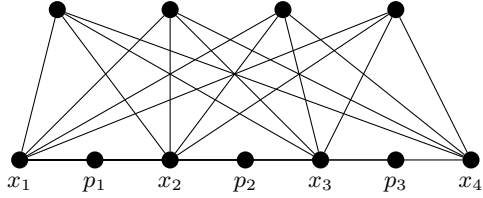
\begin{figure}
\centering
\begin{tikzpicture}[scale=1]
  \tikzstyle{vertex}=[shape=circle,draw,thick,fill=black,minimum size=2pt,inner sep=2pt]
  \node[vertex] (a) at (0.5,0) {};
  \node[vertex] (b) at (2,0)   {};
  \node[vertex] (c) at (3.5,0)  {};
  \node[vertex] (d) at (5,0)  {};
  \node[vertex,label=below:$x_1$] (x1) at (0,-2)  {};
  \node[vertex,label=below:$p_1$] (p1) at (1,-2)  {};
  \node[vertex,label=below:$x_2$] (x2) at (2,-2)  {};
  \node[vertex,label=below:$p_2$] (p2) at (3,-2)  {};
  \node[vertex,label=below:$x_3$] (x3) at (4,-2)  {};
  \node[vertex,label=below:$p_3$] (p3) at (5,-2)  {};
  \node[vertex,label=below:$x_4$] (x4) at (6,-2)  {};
\foreach \from in {a,b,c,d}
\foreach \to in {x1,x2,x3,x4}
  \path (\from) edge [-] (\to);
  \path (x1) edge[-] (p1) edge[-] (x2) edge[-] (p2) edge[-] (x3) edge[-] (p3) edge[-] (x4);
\end{tikzpicture}
\caption{Sum of degrees on chordless path $x_1, p_1,x_2, p_2, \ldots,
  x_{r-1}, p_{r-1}, x_r$ is $\Theta(m)$.}
\label{fig:secondidea}
\end{figure}

Even worse, this is still true when the initial path is already
chordless, as shown in Fig.~\ref{fig:secondidea}.  Consider the
complete bipartite clique $K_{r,r} = (V_1 \cup V_2, E_{12})$, where
$V_1 = \{ x_1, x_2, \ldots, x_r \}$. Build a new graph $G = (V,E)$
where the vertex set is $V = V_1 \cup V_2 \cup \{ p_1, \ldots,
p_{r-1}\}$ for some new vertices $p_1, \ldots, p_{r-1}$, and the edge
set is $E = E_{12} \cup \{ (x_1,p_1), (p_1,x_2), (x_2,p_2), \ldots,
(x_{r-1}, p_{r-1}), (p_{r-1}, x_r)\}$. Now, the path $x_1, p_1,x_2,
p_2, \ldots, x_{r-1}, p_{r-1}, x_r$ is chordless but each edge is
incident to at least one vertex in that path, so the sum of the
degrees is $m=|E|=\Theta(r^2) = \Theta(|V|^2) = \Theta(n^2)$.

What we would like to do: recursively extend a given chordless path
$\pi_{su}$ into a chordless $st$-path, while maintaining as a
certificate an $st$-path. The recursive extension can be seen as an
implicit cleaning of our $st$-path certificate. Consider a vertex $u$
along a given $st$-path (our certificate), where initially $u=s$. Our
certificate guarantees that there is at least one chordless $st$-path
going through a neighbour of $u$, say $a$. However, exploring all of
$u$'s neighbours would cost too much so we need to proceed more
carefully: consider any neighbour $b \neq a$, the following two
situations may occur. \emph{(1)} $a$ and $b$ are both good, meaning
that $(u,a)$ and $(u,b)$ are on two distinct chordless $st$-paths. In
this case, the chordless $st$-paths traversing $(u,a)$ cannot go
through $b$ too, as otherwise it would not be chordless (see
Remark~\ref{rem:neighbors} below), so $b$ should be
removed. \emph{(2)}~$b$ is not on any chordless $st$-path, so it is
either disconnected from $t$ or every $st$-path going through $b$
passes through $a$. In this case, as it will be clear later, we need
neither to explore nor to remove $b$.

In other words, we can treat the neighbours of $u$ as described above,
and they will not interfere when cleaning the $st$-path in the next
recursive calls since their are either removed (as in case~1) or
implicitly cut out (as in case~2). We make this statement more precise
below.

\subsection{Reduced degree property}
\label{sub:reduced-degree}
We introduce a notion of reduced degree with a stronger property in
mind.  Consider a chordless $st$-path $\pi_{st} = v_0 v_1 \dots
v_\ell$ in the graph $G$, for some integer $\ell >1$, where $v_0 = s$
and $v_\ell = t$. For a vertex $v_i$, a neighbour $v \in N(v_i)$ is
\emph{good} if there exists a chordless $st$-path in $G$ with prefix
$v_0 v_1 \dots v_i v$ (\ie it extends $v_0 v_1 \dots v_i$ by adding
the edge $(v_i,v)$ as illustrated in Fig.~\ref{fig:thirdidea}). We
denote by $\good(v_i) \subseteq N(v_i)$ the set of \emph{good}
neighbours of $v_i$, noting that $v_{i+1} \in \good(v_i)$.  For each
$v_i$, its \emph{reduced degree} $d_i$ is given by the number of
non-good neighbours, namely, $d_i = |(N(v_i) \setminus
\bigcup\limits_{j \leq i} \good(v_j)) \cup \{v_{i+1}\}|$.

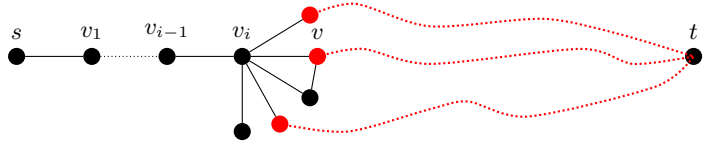
\begin{figure}[htbp]
\centering
\begin{tikzpicture}[scale=1]
  \tikzstyle{vertex}=[shape=circle,draw,thick,fill=black,minimum size=2pt,inner sep=2pt]
  \node[vertex,label=above:$s$] (s) at (0,0) {};
  \node[vertex,label=above:$v_1$] (v1) at (1,0)   {};
  \node[vertex,label=above:$v_{i-1}$] (v2) at (2,0)  {};
  \node[vertex,label=above:$v_i$] (vi) at (3,0)  {};
  \node[vertex, color=red] (a) at (3.9,0.55)  {};
  \node[vertex, color=red, label=above:$v$] (b) at (4,0)  {};
  \node[vertex, color=black] (c) at (3.9,-0.55)  {};
  \node[vertex, color=red] (d) at (3.5,-0.9)  {};
  \node[vertex, color=black] (e) at (3,-1.0)  {};
  \node[vertex,label=above:$t$] (t) at (9,0)  {};
\foreach \to in {a,b,c,d,e}
  \path (vi) edge [-] (\to); 
  
  \path (s) edge[-] (v1);
  \path (v1) edge[densely dotted] (v2);
  \path (v2) edge[-] (vi);
  \path (c) edge[-] (b);
  \draw[color=red,thick,densely dotted] plot[smooth] coordinates{(a)(4.5,0.7)(5.5,0.4)(7,0.5)(8,0.3)(t)};
  \draw[color=red,thick,densely dotted] plot[smooth] coordinates{(b)(4.6,0.1)(5.5,-0.1)(7.5,0.1)(8.2,-0.1)(t)};
  \draw[color=red,thick,densely dotted] plot[smooth] coordinates{(d)(4.4,-1)(5.5,-0.7)(6,-0.6)(6.8,-0.8)(8.5,-0.4)(t)};
\end{tikzpicture}
\caption{Good neighbours (in red) of vertex $v_i$ in $G_i$.}
\label{fig:thirdidea}
\end{figure}

The rationale is that exploring the good neighbours of $v_i$ will list
further chordless paths while examining its neighbours that are not
good is a waste of computation. The reduced degree of $v_i$ is
actually an upper bound on the number of not-good vertices examined
when exploring $v_i$ to produce the chordless $st$-path $\pi_{st}$ and
gives an upper bound on the waste. Lemma~\ref{lem:sum_degrees} below
shows that while examining the neighbours of the vertices along a
chordless path still takes $O(m)$ time, only $O(n)$ neighbours are a
waste while the remaining ones lead to further chordless paths (which
is a good argument for amortization).

\begin{lemma} 
  \label{lem:sum_degrees}
  For a chordless path $\pi_{st}$, we have $\sum_{v_i \in \pi_{st}} d_i
  \leq 2n$, where $d_i$ is the reduced degree of $v_i \in \pi_{st}$.
\end{lemma}
\begin{proof}
  We will show that each vertex $x$ of $G$ is a non-good neighbour of
  at most two vertices in $\pi_{st}$. To this purpose, we prove that
  if $x$ is a non-good neighbour of both $v_i$ and $v_j$ then $|i-j|
  \leq 1$. We choose such three vertices $v_i,v_j$ and $x$ where the
  difference $j-i$ is the largest possible and assume by contradiction
  that $i < j-1$. Thus $v_i$ and $v_j$ are not adjacent in $\pi_{st}$,
  whence $(v_i,v_j)$ is not an edge of $G$ since $\pi_{st}$ is
  chordless. Also, being non-good, $x \notin \pi_{st}$. Consider the
  $st$-path $\pi^* = v_0 \ldots v_i x v_j \ldots v_l$. Clearly,
  $\pi^*$ contains no repeated vertices and we will prove that $\pi^*$
  is a chordless $st$-path, contradicting the fact that $x$ is not a
  good neighbour of $v_i$. The fact that $\pi^*$ is chordless follows
  from the fact that there is no $v_k \in \pi^*$, $k \neq i$ and $k
  \neq j$, such that $(v_k, x)$ is an edge of $G$, otherwise $j - k$
  or $k - i$ would be strictly larger than $j-i$, contradicting our
  choice of $v_i,v_j$ and $x$. \qed
\end{proof}

\subsection{Cleanup of current vertex}
\label{sub:cleanup-operation}
Suppose we are extending the chordless path $\pi_{su}$, while cleaning
the $st$-path certificate. We identify a good vertex $v \in N(u)$,
which closest to $t$ along the $st$-path.  Ideally, we would clean the
vertex $u$ by throwing away all its other neighbours but this could
cost $\Omega(m)$ per chordless path as illustrated in
Fig.~\ref{fig:secondidea} and discussed in
Section~\ref{sub:cleaning}. We thus perform a partial cleaning, called
\emph{cleanup}, which consists in identifying and removing, among all
neighbours of $u$ (\textit{i.e.}  $|N(u)|$ elements) only its set
$\good(u)$ of good ones.

For a given $u$ in a chordless $st$-path $\pi_{st} = v_0 \ldots v_i u
\ldots v_l$, we let emerge the good neighbours in $\good(u)$ one by
one as follows.  Consider the graph $G'$ where the vertices $v_0
\ldots v_i$ and its good neighbours were removed. If $u$ and $t$ are
not connected, then there cannot be further chordless paths from $u$
and so there cannot be further good neighbours.  Otherwise, if $u$ and
$t$ are connected, we take any path from $u$ to $t$, and select its
neighbour $v$ that appears along the path and is closest to $t$, as
illustrated in Fig.~\ref{fig:fourthidea}. After that, we remove $v$
and its incident edges, and iterate what described above until $u$ is
disconnected from $t$. The vertices $v$ thus selected form the set
$\good(u)$ of good neighbours.

\begin{figure} 
\centering
\begin{tikzpicture}[scale=1]
  \tikzstyle{vertex}=[shape=circle,draw,thick,fill=black,minimum size=2pt,inner sep=2pt]
  \node[vertex,label=above:$s$,color=lightgray] (s) at (0,0) {};
  \node[vertex,label=above:$v_1$,color=lightgray] (v1) at (1,0)   {};
  \node[vertex,label=above:$v_{i-1}$,color=lightgray] (v2) at (2,0)  {};
  \node[vertex,label=above:$v_i$,color=lightgray] (vi) at (3,0)  {};
  \node[vertex,color=red, label=above left:$u$] (u) at (4,0)  {};
  \node[vertex] (a) at (5,0)  {};
  \node[vertex] (b) at (6,0)  {};
  \node[vertex] (c) at (4.3,0.7)  {};
  \node[vertex] (d) at (4.7,-0.5)  {};
  \node[vertex,label=above:$v$] (v) at (7,0)  {};
  \node[vertex,label=above:$t$] (t) at (9,0)  {};
 
  \path (s) edge[-,color=lightgray] (v1);
  \path (v1) edge[densely dotted,color=lightgray] (v2);
  \path (v2) edge[-,color=lightgray] (vi);
  \path (vi) edge[-,color=lightgray] (u);
  \path (u) edge[-, bend left=40, color=red] (b);
  \path (u) edge[-, thick, color=red] (a);
  \path (u) edge[-, thick, color=red] (c);
  \path (u) edge[-, thick, color=red] (d);
  \path (u) edge[-, bend right=80, thick, color=blue] (v);
  \draw[color=red,thick,densely dotted] plot[smooth]
  coordinates{(a)(5.5,0.1)(b)(6.5,-0.1)(v)};
  \path (v) edge[-, thick, color=red] node[near start,below, color=black]{} (t);
  \path (v) edge[-, thick, dashed, color=blue] node[near start,below, color=black] {$\pi_{v,t}$} (t);
\end{tikzpicture}
\caption{Cleanup of the neighbours of vertex $u$.}
\label{fig:fourthidea}
\end{figure}
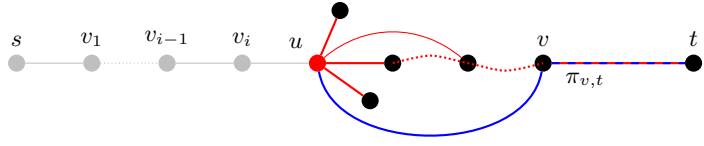

\begin{lemma}
  \label{lemma:cleanup}
  For a chordless path $\pi_{st}$, the cleanup of vertex $u \in
  \pi_{st}$ correctly produces the set $\good(u)$ of its good
  neighbours.
\end{lemma}
\begin{proof}
  Given a chordless $st$-path $\pi_{st} = v_0 v_iu \dots v_\ell$ in
  the graph $G$, for some integer $\ell >1$, where $v_0 = s$ and
  $v_\ell = t$.  Let $G'$ be the subgraph of $G$ where the vertices
  $\{v_0, \ldots, v_i\}$ and their good neighbours were removed. Let
  $S \subseteq N(v_i)$ be the set of vertices removed by the cleanup
  procedure for $u$ in $G'$. We divide the proof in two parts. We
  first show that $S$ contains $\good(v_i)$, and then show that
  $\good(v_i)$ contains $S$.

  Clearly, for the first part, it is enough to show that $N_{G'}(v_i)
  \setminus S$ are not good neighbours. The vertices of $N_{G'}(v_i)
  \setminus S$ cannot reach $t$ in $G'$ without passing through some
  vertex in $S$, otherwise $u$ would still be connected to $t$ and the
  procedure would not stop. This implies that there is no chordless
  $ut$-path in $G'$ using some vertex of $N_{G'}(v_i) \setminus S$,
  \textit{i.e.} they are not good neighbours of $v_i$.

  Finally, for the second part, it is enough to show that for each $w
  \in S$ there is a chordless $ut$-path in $G'$ passing through
  $w$. Consider the iteration where $w \in N_{G'}(u)$ was added to $S$
  and let $S' \subseteq S$ be the set of vertices added in previous
  iterations. We have that $w$ is the vertex closest to $t$ in the path
  $\pi = u \leadsto t$ in $G' \setminus S'$. We claim that any subpath
  of $(u,w)\pi_{wt}$ in $G' \setminus S'$ contains $(u,w)$, where
  $\pi_{wt} = w \leadsto t$ is a suffix of $\pi$. Thus implying that
  $w$ is contained in an induced $ut$-path in $G'$. Indeed, the edge
  $(u,w)$ is not contained in a subpath iff there is a vertex $x \in
  N_{G'}(u)$ in $\pi_{wt}$. By construction, $\pi_{wt}$ does not
  contain any vertex of $S'$; and by the choice of $w$, $\pi_{wt}$
  does not contain any vertex of $N_{G'}(u) \setminus S'$. \qed
\end{proof}

\section{Listing Algorithm}
\label{sec:listing-algorithm}

We blend the key ideas discussed in
Section~\ref{sec:approach-key-ideas} to get
Algorithm~\ref{alg:improved}, which has four parameters as input and
lists all the chordless $st$-paths: the first parameter is the
chordless path $\pi_{su}$ partially built from $s$ to the current
vertex $u$ (initially, $u=s$), which is the second parameter; the
third parameter is a $ut$-path $\pi_{ut}$ that plays the role of
certificate by Fact~\ref{fact:path_induced}; the fourth parameter is
the reduced graph $G$, which changes with the recursive calls.

\begin{algorithm} 
\caption{$\liststpaths(\pi_{su}, u, \pi_{ut}, G)$} \label{alg:improved}
\uIf{$u = t$}{ 
  output($\pi_{su}$) \\ \label{alg:output}
}\Else{
  $S := \emptyset$ \\ \label{alg:init}
  \While{true}{ \label{alg:loop:begin}
    $v := $ the vertex in $\pi_{ut} \cap N(u)$ that is closest to $t$ in $\pi_{ut}$ \\ \label{alg:last_vertex}
    $\pi_{vt} := $ the subpath of $\pi_{ut}$ from $v$ to $t$\\   \label{alg:path_suffix}
    $S := S \cup \{(v,\pi_{vt})\}$ \\
    remove $v$ and its incident edges from $G$ \\ \label{alg:update1}
    \lIf{$u$ and $t$ are not connected}{{\bf break}}\\ \label{alg:query}
    $\pi_{ut} := $ any path from $u$ to $t$ \\  \label{alg:path}
  } \label{alg:loop:end}
  \ForEach{$(v, \pi_{vt}) \in S$}{ \label{alg:recursion:begin}
    adds back $v$ and its incident edges to $G$ \\ \label{alg:add}
    $\liststpaths(\pi_{su} \cdot (u,v), v, \pi_{vt}, G)$ \\ \label{alg:recursive_call}
    remove $v$ and its incident edges from $G$ \\ \label{alg:remove}
  } \label{alg:recursion:end}
} 
\end{algorithm}

The algorithm outputs a chordless $st$-path if $u=t$
(line~\ref{alg:output}). Otherwise, it performs a cleanup of $u$ (the
loop at lines~\ref{alg:loop:begin}--\ref{alg:loop:end}). After that,
it explores only the good neighbours recursively as they will surely
lead to further chordless paths (the other loop at
lines~\ref{alg:recursion:begin}--\ref{alg:recursion:end}). Observe
that $S$ stores the good neighbours $v$ of $u$ and a $vt$-path for
each of them: when performing the recursive call at
line~\ref{alg:recursive_call}, only one of the vertices in $S$ appears
in the reduced graph $G$ passed as a parameter to the recursive call
(see lines~\ref{alg:add} and~\ref{alg:remove} that guarantee this, and
Remark~\ref{rem:neighbors} below). Hence, the recursive call now has
as parameters the chordless $sv$-path $\pi_{su} \cdot (u,v)$ ending in
$v$, and a $vt$-path that guarantees that a chordless $st$-path exists
and has $\pi_{su} \cdot (u,v)$ as a prefix. This recursive call lists
all the chordless $st$-paths that share this prefix.

\begin{figure}
\centering
\begin{tikzpicture}[scale=1]
  \tikzstyle{vertex}=[shape=circle,draw,thick,fill=black,minimum size=2pt,inner sep=2pt]
  \node[vertex,label=above:$v_0$] (v0) at (0,0) {};
  \node[vertex,label=above:$v_1$] (v1) at (1,1)   {};
  \node[vertex,label=above:$v_2$] (v2) at (2,1)  {};
  \node[vertex,label=above:$v_3$] (v3) at (3,0)  {};
  \node[vertex,label=above:$v_4$] (v4) at (4,0)  {};
  \node[vertex,label=below:$v_5$] (v5) at (1,-1)  {};
  \node[vertex,label=below:$v_6$] (v6) at (2,-1)  {};
  \path (v0) edge [-, bend left=10] (v1);
  \path (v0) edge [-] (v3);
  \path (v0) edge [-, bend right=10] (v5);
  \path (v0) edge [-, bend left=10] (v6);
  \path (v1) edge [-] (v2);
  \path (v2) edge [-, bend left=10] (v3);
  \path (v3) edge [-] (v4);
  \path (v3) edge [-, bend left=10] (v6);
  \path (v4) edge [-, bend left=25] (v6);
  \path (v5) edge [-] (v6);
\end{tikzpicture}~~~~~~~~~~~~\begin{tikzpicture}[scale=1]
  \tikzstyle{vertex}=[shape=circle,draw,thick,fill=black,minimum size=2pt,inner sep=2pt]
  \node[vertex,label=above:$v_0$] (v0) at (0,0) {};
  \node[vertex,label=above:$v_2$] (v2) at (3,0)  {};
  \node[vertex,label=above:$v_4$] (v4) at (4,0)  {};
  \node[vertex,label=below:$v_1$] (v1) at (2,-1)  {};
  \node[vertex,label=below:$v_3$] (v3) at (3,-1)  {};
  \node[vertex,label=below:$v_5$] (v5) at (4,-1)  {};
  \path (v0) edge [-] (v2);
  \path (v2) edge [-] (v3);
  \path (v2) edge [-] (v4);
  \path (v0) edge [-, bend right=10] (v1);
  \path (v1) edge [-] (v3);
  \path (v3) edge [-] (v5);
  \path (v4) edge [-] (v5);
\end{tikzpicture}
\caption{Two example graphs where $s=v_0$ and $t=v_4$.}
\label{fig:example}
\end{figure}
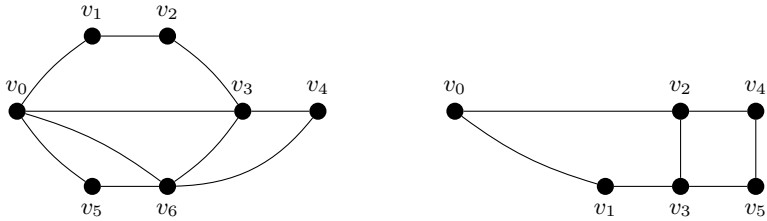

For example, let us run Algorithm~\ref{alg:improved} on the input
graph shown on the left of Fig.~\ref{fig:example}, with $u=s=v_0$ and
the initial path $\pi_{ut} = v_0v_1v_2v_3v_4$.  It computes the pairs
$(v, \pi_{vt})$ in $S$ as follows. First, $(v_3, v_3v_4)$ is added to
$S$ as $v_3$ is a good neighbour for $v_0$ (the neighbour closest to
$t$ in the path), and the edges incident to $v_3$ are removed. After
this removal, $s=v_0$ is still connected to $t=v_4$ through the path
$v_0v_5v_6v_4$, which becomes the input for the next iteration of the
while loop. Next, $(v_6, v_6v_4)$ is added to $S$ as $v_6$ is another
good neighbour, and the edges incident to $v_6$ are removed
disconnecting $s$ from $t$, so the while loop ends.  The recursive
calls in the foreach loop give the two chordless paths $v_0v_3v_4$ and
$v_0v_6v_4$ contained in the graph.

\begin{remark}
  \label{rem:neighbors}
  It is important to run the recursive calls with all good neighbours
  in $S$ removed except one. If we left two or more good neighbours in
  the recursive call of line~\ref{alg:recursive_call}, they could
  interfere with each other and we might not obtain the chordless paths
  correctly. A very simple example is given in the graph shown on the
  right of Fig.~\ref{fig:example}. Consider for instance the case
  where Algorithm~\ref{alg:improved} would be given as input the path
  $v_0 v_1 v_3 v_2 v_4$. The pair $(v_2, v_2v_4)$ is added to $S$ and
  $v_2$ removed. After that, the path $v_0 v_1 v_3 v_5 v_4$ is found
  and the pair $(v_1, v_1v_3v_5v_4)$ is added to~$S$ and $v_1$
  removed. Since $v_0$ and $v_4$ become disconnected, $S$ contains all
  the good neighbours of $v_0$. Algorithm~\ref{alg:improved} executes
  the recursive calls with $S$.  Suppose that we keep both good
  neighbours $v_1$ and $v_2$ in $G$ during these calls, in particular
  for the call with the pair $(v_1, v_1v_3v_5v_4)$ from $S$. This call
  will extend in a nested call the chordless path to $\pi_{su} = v_0 v_1
  v_3$ for $u=v_3$, and will claim that the good neighbours of $v_3$
  are $v_2$ and $v_5$, which is incorrect since $v_0 v_1 v_3 v_2$ is
  not chordless. This situation does not arise if $v_2$ is kept deleted
  in $G$ when the recursive call on $v_1$ is performed as done in
  Algorithm~\ref{alg:improved}.
\end{remark}

The correctness of Algorithm~\ref{alg:improved} follows mostly from
Lemma~\ref{lemma:cleanup}. Recall that, it guarantees that for a given
path prefix $\pi_{su}$ the set $S$ contains the good neighbours of
$u$, \textit{i.e.} the neighbours of $u$ that belong to at least one
chordless $st$-path extending $\pi_{su}$. Clearly, we only have to
recursively call the algorithm for these neighbours, the others
certainly lead to no solution. This implies that
Algorithm~\ref{alg:improved} tries all the possibilities to extend
$\pi_{su}$, so all chordless $st$-paths are output. Moreover, since
each good neighbour of $u$ leads to a different extension, we have
that no $st$-path is output more than once.

Certainly only $st$-paths are output by Algorithm~\ref{alg:improved},
but at this point we have no guarantees that the paths are indeed
chordless. In fact, after building $S$, the algorithm proceeds to
recursively extend the prefix $\pi_{su}\cdot (u,v)$ for each $v \in S$
in the graph $G' = G \setminus (S \setminus \{v\})$. However, since
$u$ was included in current path none of its neighbours can be used
later in the recursion. The algorithm removes the good neighbours of
$u$ from $G$, but the other neighbours, $N_G(u) \setminus S$, are
still present in $G'$. They could thus be used to extend the path
later in the recursion, resulting in a non-chordless
$st$-path. Lemma~\ref{lem:only_induced} shows that this cannot happen.

\begin{lemma} \label{lem:only_induced}
  The $st$-paths output by Algorithm~\ref{alg:improved} are chordless.
\end{lemma}
\begin{proof}
  We proceed by contradiction. Suppose $\pi_{st}$ is output by the
  algorithm and it is not chordless. This means that there exists a pair
  of vertices $x,y \in \pi_{st}$ such that $x \neq y$, $(x,y)$ is an
  edge of $G$, and $(x,y) \notin \pi_{st}$.  We can assume the edge
  $(x,y)$ is such that $x$ is the vertex closest to $s$. Now, consider
  the recursive call corresponding to the prefix $\pi_{sx}$: let $G_x$
  be the associated graph and $z$ the next vertex of $\pi_{st}$.  The
  suffix $\pi_{xt}$ of $\pi_{st}$ must pass through $z$, since $x$ is
  not connected to $t$ in $G_x - S$.  So, $y$ is closest to $t$ than
  $z$ in $\pi_{st}$. Thus, the path $(x,y)\pi_{yt}$, where $\pi_{yt}$
  is a suffix of $\pi_{st}$, avoids $S$ in $G_x$. This contradicts the
  test in line~\ref{alg:query}. \qed
\end{proof}

The previous lemma leads to the following theorem.

\begin{theorem}
  \label{the:correctness}
  The algorithm correctly outputs all chordless $st$-paths of $G$.
\end{theorem}

\begin{theorem}
  \label{the:complexity}
  The algorithm takes $O(m + |\setofinducedpaths_{st}(G)|(t_p + nt_q +
  nt_u))$ time, where $t_p$ is the cost of choosing any path from any
  two given vertices, $t_q$ is the cost of checking if any given two
  vertices are connected or not, and $t_u$ is the cost of
  removing/adding back any given edge.
\end{theorem}
\begin{proof}
See Section~\ref{sec:amortized-analysis}.
\end{proof}

There are several dynamic data structures in the literature
\cite{Kapron2013} that maintain a spanning forest for a dynamic graph,
supporting insertions and deletions of edges in polylogarithmic
time. Consequently, $t_p = O(n \,\polylog(n))$, $t_q =
O(\polylog(n))$, and $t_u = O(\polylog(n))$, thus giving the following
bound.

\begin{corollary}
  The algorithm takes $\tilde{O}(m + |\setofinducedpaths_{st}(G)| \cdot
  n)$ time to report all the chordless $st$-paths.
\end{corollary}

\section{Amortized Analysis}
\label{sec:amortized-analysis}

Before starting our analysis, we observe some simple properties of the
recursion tree generated by Algorithm~\ref{alg:improved}.

\begin{fact} 
  \label{fact:tree_prop}
  The recursion tree $R$ of Algorithm~\ref{alg:improved} has the
  following properties:
  \begin{enumerate}
    \item There is a one-to-one correspondence between paths in
      $\setofinducedpaths_{st}(G)$ and leaves in the recursion
      tree. \label{item:R1}
    \item There is a one-to-one correspondence between proper prefixes
      of paths in $\setofinducedpaths_{st}(G)$ and internal nodes in
      the recursion tree. \label{item:R2}
    \item The number of branching nodes is $|\setofinducedpaths_{st}(G)|-1$. \label{item:R3}
    \item The length of a root-to-leaf path is equal to the length of
      the chordless $st$-path corresponding to the leaf. In particular,
      the height of the tree is $\leq n$.
  \end{enumerate}
\end{fact}

\noindent Fact~\ref{fact:tree_prop} suggests us to follow the following  overall strategy.
\begin{enumerate}
  \item We analyze the cost of each type (leaf, unary and branching)
    of node separately.
  \item We consider all branching nodes together, and show that their
    amortized cost is $O(t_p + t_q + nt_u + n) = \tilde{O}(n)$ per
    solution.
  \item We consider all unary nodes together, and show that their
    amortized cost is $O(|\pi_{st}|t_q + nt_u) = \tilde{O}(n)$ per
    solution.
  \item We deduce that the cost of each solution is $O(t_p + nt_q +
    nt_u) = \tilde{O}(n)$.
\end{enumerate}

\noindent Where the cost of a node is the time spent by the
corresponding call without including the time spent by its nested
recursive calls.

\begin{lemma} \label{lem:cost_leaf}
  The cost of a leaf is $O(|\pi_{st}|)$.
\end{lemma}
\begin{proof}
  Clearly, when $u = t$, the only operation done by the algorithm is
  to output $\pi_{st}$, which takes $O(|\pi_{st}|)$ time. \qed
\end{proof}

Let us now analyze the cost of the unary nodes. Let $r = \langle
\pi_{su}, u, \pi_{ut}, G \rangle$ be a unary node. The vertex $v \in
N(u)$ is the only neighbour of $u$ that can extend the prefix
$\pi_{su}$ into a chordless $st$-path. Thus, removing $v$ from $G$
disconnects $u$ from $t$, and the algorithm performs a single
iteration of the loop in line~\ref{alg:loop:begin}, not executing
line~\ref{alg:path}. In this case, the algorithm performs the
following operations: (i) one connectivity query
(line~\ref{alg:query}), (ii) $|N(v)|$ edge update operations on $G$
(lines~\ref{alg:update1}, \ref{alg:add} and \ref{alg:remove}), and
(iii) a scan in the intersection of $N(u)$ and $\pi_{ut}$ to find $v$
(line~\ref{alg:last_vertex}). The cost of (i) and (ii) is $O(t_q +
|N(v)|t_u)$.

A naive implementation of (iii) takes $O(|N(u)| + |\pi_{ut}|)$ time,
which is too large to fit in our amortization strategy. In order to
reduce this cost to $O(|N(u)|)$ we therefore maintain, as an extra
invariant, for each vertex in the current graph its distance to $t$
along the path $\pi_{ut}$. In this way, we can find $v$ simply
scanning $N(u)$. Thus, assuming the distance information is correctly
maintained, we complete the proof of Lemma~\ref{lem:cost_unary}.

\begin{lemma} \label{lem:cost_unary}
  The cost of a unary node is $O(t_q + |N(v)|t_u + |N(u)|)$, where
  $(u,v)$ is the edge added to the chordless path.
\end{lemma}

It is not hard to maintain the distance information for $\langle
\pi_{su}\cdot (u,v), v, \pi_{vt}, G'\rangle$, the only child of the
unary node $\langle \pi_{su}, u, \pi_{ut}, G \rangle$. As the path
$\pi_{vt}$ is a suffix of $\pi_{ut}$, the distance of the vertices in
$\pi_{ut}$ does not change. On the other hand, the only vertices that
the distances can change are the ones in $\pi_{vt}$ but not in
$\pi_{ut}$.  These vertices can be identified when scanning $N(v)$ in
the child node $\langle \pi_{su}\cdot (u,v), v, \pi_{vt}, G'\rangle$,
since their distance is strictly larger than $|\pi_{vt}|$. It remains
to show that the distance information can be maintained in the
branching nodes.

\begin{lemma} \label{lem:cost_branching}
  The cost of a branching node $r \in R$ is $O(\beta(r)(t_p + t_q +
  nt_u))$, where $\beta(r)$ is the number of children of $r$.
\end{lemma}
\begin{proof}
  The cost of a branching node $r = \langle \pi_{su}, u, \pi_{ut}, G
  \rangle$ is dominated by the cost of the loop of
  line~\ref{alg:loop:begin}. The number of iterations of the loop is
  equal to the number of neighbours of $u$ that can extend $\pi_{su}$
  into a chordless $st$-path, which is exactly the number of vertices
  in $S$ after the loop finishes, \ie $\beta(r)$, the number of
  children of $r$. Let us now bound the cost of each iteration. The
  cost of lines~\ref{alg:last_vertex} and \ref{alg:path_suffix} is
  bounded by $O(|N(u)| + |\pi_{ut}|) = O(n)$: we simply have to
  traverse the path $\pi_{ut}$ and scan the set $N(u)$. The cost of
  updating $G$ is bounded by $O(nt_u)$. Finally, the cost for the
  connectivity query and to find a path is $t_p + t_q$. Hence, the
  total cost for a branching node is $O(\beta(r)(t_p + t_q +
  nt_u))$. \qed
\end{proof}

Let us now show that we can maintain the distance information in
branching nodes in the same time bound of
Lemma~\ref{lem:cost_branching}. This follows from the fact that in
each iteration of the loop (line~\ref{alg:loop:begin}) we are already
paying $O(|\pi_{ut}|)$, \ie a full traversal of the path $\pi_{ut}$.
Before each recursive call in line~\ref{alg:recursive_call} we can
traverse the path $\pi_{ut}$ adding for each vertex the distance
information, \ie their position in the path.

At this point we have bounds for the cost of each node in the
recursion tree. However, by directly applying them we cannot achieve
our goal of $\tilde{O}(n)$ time per solution. For instance, consider
the particular case where all internal nodes of the recursion tree are
branching. The cost of each internal node is $O(\beta(r)(t_p + t_q +
nt_u)) = \tilde{O}(n^2)$, since $\beta(r) = \Omega(n)$ in the worst
case. Then, from item~\ref{item:R3} of Fact~\ref{fact:tree_prop}, the
number of branching nodes is $|\setofinducedpaths_{st}(G)|-1$.  The
total cost for the tree is thus
$\tilde{O}(|\setofinducedpaths_{st}(G)|\cdot n^2)$ or $\tilde{O}(n^2)$ per
solution.

In order to get a tighter bound for the total cost of the branching
nodes, we use the following amortization strategy. Let $r \in R$ be a
branching node. We divide the cost $O(\beta(r)(t_p + t_q + nt_u))$
among the closest descendents that are branching nodes or leaves (no
unary nodes), each being charged $O(t_p + t_q + nt_u)$. This can
always be done since $r$ has $\beta(r)$ children and the subtree of
each child contains at least one leaf, \ie the node $r$ has at least
$\beta(r)$ non-unary descendants. In this way, the original cost of
node $r$ is completely charged to its non-unary descendants, and the
only cost that remains associated to $r$ is the one received from its
ancestors. Finally, each branching node can only be charged once, by
its lowest non-unary ancestor. Each branching node and each leaf is
therefore charged with $O(t_p + t_q + nt_u)$. Thus, the total cost of
the branching nodes is $O(|\setofinducedpaths_{st}(G)|(t_p + t_q +
nt_u))$, completing the proof of Lemma~\ref{lem:branching_total_cost}.

\begin{lemma} \label{lem:branching_total_cost}
  $\sum_{r: \mathrm{branching}} T(r) =
  O(|\setofinducedpaths_{st}(G)|(t_p + t_q + nt_u))$.
\end{lemma}

Let us now bound the total cost of the unary nodes. Similarly to the
branching nodes case, a straightforward use of the bound given by
Lemma~\ref{lem:cost_unary} leads to an $\tilde{O}(n^2)$ cost per
solution, since in the worst case the recursion tree can have $O(n)$
unary nodes for each leaf. The key idea to obtain a better amortized
cost is to consider the bound on the reduced degrees given by
Lemma~\ref{lem:sum_degrees}.

We first observe that each unary node is contained in some
root-to-leaf path $\Pi(l)$, where $l$ is a leaf of the recursion
tree. Thus,
\begin{equation}
  \sum_{r: \mathrm{unary}} T(r) \leq \sum_{l: \mathrm{leaf}} \sum_{r \in \Pi(l)} T(r).
\end{equation}

Fact~\ref{fact:tree_prop} implies that there is a one-to-one
correspondence between the prefixes of paths in
$\setofinducedpaths_{st}(G)$ and nodes in the recursion tree. That is,
each leaf corresponds to a solution, and the root-to-leaf path
$\Pi(l)$ corresponds to the chordless $st$-path associated to the leaf
$l$. Moreover, the $O(t_q + |N(v)|t_u + |N(u)|)$ cost of an unary node
can be amortized to $O(t_q + |N(v)|t_u)$, since we can always charge
$|N(u)| = O(n)$ to its single child. We can thus rewrite the double
sum as
\begin{equation} \label{eq:sum_unary_cost}
  \sum_{l: \mathrm{leaf}} \sum_{r \in \Pi(l)} T(r) = 
  \sum_{\pi \in \setofinducedpaths_{st}(G)} \sum_{v_i \in \pi} (t_q +
  |N(v_i)|t_u). 
\end{equation}

For each chordless $st$-path $\pi$ in the internal sum of
Eq.~\ref{eq:sum_unary_cost}, we have that the degrees are actually the
reduced degrees of Section~\ref{sub:reduced-degree}, since the good
neighbours (\ie the set $S$ in Algorithm~\ref{alg:improved}) are
always removed. Using Lemma~\ref{lem:sum_degrees} we can thus bound
the sum of the degrees by $2n$. Therefore,
\begin{equation}
  \sum_{r: \mathrm{unary}} T(r) \leq \sum_{\pi \in
    \setofinducedpaths_{st}(G)} (|\pi|t_q + 2nt_u),
\end{equation}
completing the proof of Lemma~\ref{lem:unary_total_cost}.

\begin{lemma} \label{lem:unary_total_cost}
  $\sum_{r: \mathrm{unary}} T(r) = O(\sum_{\pi \in \setofinducedpaths_{st}(G)}
  |\pi|t_q + nt_u)$.
\end{lemma}

As a corollary of Lemmas~\ref{lem:unary_total_cost} and
\ref{lem:branching_total_cost}, we obtain Theorem~\ref{the:complexity}.

\end{document}